\documentclass[12pt]{article}

\usepackage{times}

\usepackage{amsfonts,amssymb,amsmath,amsthm}
\usepackage{latexsym}
\usepackage{url}
\usepackage[breaklinks]{hyperref}

\def\01{\{0,1\}}

\newcommand{\rk}{\mathrm{rk}}

\newcommand{\ignore}[1]{}

\newtheorem{theorem}{Theorem}

\newtheorem{claim}[theorem]{Claim}

\newtheorem{definition}[theorem]{Definition}
\newtheorem{fact}[theorem]{Fact}

\newcommand{\reff}[1]{\ref{#1}}
\newcommand{\thmref}[1]{\hyperref[#1]{{Theorem~\reff{#1}}}}
\newcommand{\lemref}[1]{\hyperref[#1]{{Lemma~\reff{#1}}}}
\newcommand{\corref}[1]{\hyperref[#1]{{Corollary~\reff{#1}}}}
\newcommand{\eqnref}[1]{\hyperref[#1]{{Equation~(\reff{#1})}}}
\newcommand{\factref}[1]{\hyperref[#1]{{Fact~\reff{#1}}}}
\newcommand{\defref}[1]{\hyperref[#1]{{Definition~\reff{#1}}}}
\newcommand{\secref}[1]{\hyperref[#1]{{Section~\reff{#1}}}}
\newcommand{\claimref}[1]{\hyperref[#1]{{Claim~\reff{#1}}}}

\begin{document}
\title{Rank and fooling set size}
\author{Aya Hamed \thanks{Work done in part while visiting the Centre for Quantum Technologies, Singapore} \and Troy Lee \thanks{Nanyang Technological University and Centre for Quantum 
Technologies, Singapore.  Research supported by a National Research Foundation Fellowship.}}
\maketitle

\begin{abstract}
Say that $A$ is a Hadamard factorization of the identity $I_n$ of size $n$ if $A \circ A^T = I_n$, where 
$\circ$ denotes the 
Hadamard or entrywise product, and $A^T$ denotes the transpose of $A$.  As $n=\rk(I_n)=\rk(A \circ A^T) \le \rk(A)^2$, it is clear that the rank of any Hadamard 
factorization of the identity must be at least $\sqrt{n}$.  Dietzfelbinger et al. \cite{DHS96} raised the question if 
this bound can be achieved, and showed a boolean Hadamard factorization of the identity with 
$\rk(A) \le n^{0.792\ldots}$.  
More recently, Klauck and Wolf \cite{KlauckWolf13} gave a construction of Hadamard factorizations of the identity 
of rank $n^{0.613\ldots}$.  Over finite fields, Friesen and Theis \cite{FriesenTheis13} resolved the question, 
showing for a prime $p$ and $r=p^t+1$ a Hadamard factorization of the identity $A$ of size $r(r-1)+1$ and 
$\rk_p(A)=r$, where $\rk_p(\cdot)$ indicates the rank over $\mathbb{F}_p$.  

Here we resolve the question for fields of zero characteristic, up to a constant factor, giving a construction of 
Hadamard factorizations of the identity of rank $r$ and size $\binom{r+1}{2}$.  The matrices in our construction are 
blockwise Toeplitz, and have entries whose magnitudes are binomial coefficients.  
\end{abstract}

\section{Introduction}
A fooling set for a matrix $M$ of size $n$ is a set of pairs $(i_1, j_1), \ldots, (i_n, j_n)$ such that $M(i_k,j_k) \ne 0$ 
for all $k=1, \ldots, n$ yet $M(i_k, j_\ell) M(i_\ell, j_k)=0$ for any $k \ne \ell$.  Fooling sets have been studied 
in several contexts, for example in communication complexity where they provide a lower bound on nondeterministic 
communication complexity (see \cite{KushilevitzNisan97}), and also in combinatorial optimization as a lower bound 
technique for extended formulation size \cite{Yannakakis91}.  

Notice that if $M$ has a fooling set of size $n$ then $M$ contains a submatrix $A$ of size $n$ such that 
$A \circ A^T$ is nonzero on the diagonal and zero in all off-diagonal entries.  Call $A$ a Hadamard factorization of the 
identity of size $n$ if $A \circ A^T=I_n$.  As $n=\rk(I_n)=\rk(A \circ A^T) \le \rk(A)^2$, it is clear that the rank of any 
Hadamard factorization of the identity must be at least $\sqrt{n}$.  Matrix rank is another common lower bound 
technique in communication complexity, and Dietzfelbinger et al.\ \cite{DHS96}  
investigated the question of how these measures compare.  They gave a construction of a family of boolean matrices 
$A_n$ that are Hadamard factorizations of the identity of size $n$ and of rank $n^{0.792\ldots}$, and asked the 
question if a quadratic separation can indeed be achieved. 

Klauck and Wolf \cite{KlauckWolf13} took up this question in a slightly different context.  For a boolean matrix $M$, 
let the 
nondeterministic rank of $M$ be the minimum rank of a matrix $A$ such that $A(i,j)=0$ whenever $M(i,j)=0$ and 
$A(i,j) \ne 0$ whenever $M(i,j)=1$.  Notice that the question of the relationship between the nondeterministic rank of 
$M$ and fooling set size is exactly the question of the rank of Hadamard factorizations of the identity.  Klauck and Wolf 
showed a family $A_n$ of Hadamard factorizations of the identity of size $n$ and $\rk(A_n)=n^{0.613\ldots}$, and 
again raised the question if there is a construction achieving exponent $1/2$.

Friesen and Theis \cite{FriesenTheis13} recently resolved the question of rank versus fooling set size over finite fields.  
For any prime 
$p$ and $r=p^t+1$ they give a Hadamard factorization of the identity $A$ of size $r(r-1)+1$ and $\rk_p(A)=r$, where 
$\rk_p(\cdot)$ denotes the rank over $\mathbb{F}_p$.  They ask if a similar result holds over fields of zero characteristic.  

We answer this question up to constant factors.  We give a construction of a Hadamard factorization of the identity 
of rank $r$ and size $\binom{r+1}{2}$.  The entries of our matrix are integers, and in fact binomial coefficients, up 
to sign.  The construction of Friesen and Theis also has entries that are binomial coefficients (mod $p$) and, as 
they do, we show the low rank property using the binomial addition identity.  While the matrix of Friesen and Theis is 
circulant, our matrix has a more complicated block structure but each block is a Toeplitz matrix. 

\section{Preliminaries}
We will make use of binomial coefficients and a few of their standard properties.  As multiple extensions of 
binomial coefficients to negative arguments are possible, we fix here the definition we use (following \cite{KGP94}).

\begin{definition}[Binomial Coefficients]
For integer $n,k$ define
\[
\binom{n}{k} = \begin{cases}
\frac{n(n-1) \cdots (n-k+1)}{k(k-1) \cdots 1} & \text{integer } k \ge 0 \\
0 & \text{integer } k <0 \enspace .
\end{cases}
\]
\end{definition}

\begin{fact}[Symmetry]
For any $n \ge 0$ and integer $k$
\[
\binom{n}{k}=\binom{n}{n-k} \enspace .
\]
\end{fact}

\begin{fact}[Addition Formula]
For any integer $n,k$
\[
\binom{n}{k}=\binom{n-1}{k} + \binom{n-1}{k-1} \enspace .
\]
\end{fact}

\section{Construction}
We now give a construction of a family of matrices $A_r$ of rank $r$ and size $\binom{r+1}{2}$, for any integer 
$r \ge 1$.  To get some feeling for these matrices, here are the first few examples
\[
A_1=\begin{bmatrix}
1
\end{bmatrix}, 
A_2 = \begin{bmatrix}
1 & 0 & 1\\
-1 & 1 & 0 \\
0 & 1 & 1\\
\end{bmatrix},
A_3 = \begin{bmatrix}
1 & 0 & 0 & 1 & -1 & 1\\
-1 & 1 & 0 & 0 & 1 & 0\\
1 & -1 & 1 & -1 & 0 & 0\\
0 & 1 & 0 & 1 & 0 & 1\\
0 & 0 & 1 & -1 & 1 & 0\\
0 & 1 & 1 & 0 & 1 & 1
\end{bmatrix} \enspace .
\] 

The recursive structure of $A_r$ can be seen from these examples \footnote{If the reader wants to see larger 
examples, Matlab code to construct $A_r$ can be found at \url{https://github.com/troyjlee/hadamard_factorization}}.  
In general, the top left $r$-by-$r$ principal 
submatrix of $A_r$ will be lower triangular with ones of alternating sign, and the bottom right $\binom{r}{2}$-sized
principal submatrix will be $A_{r-1}$.  

We now give the details of the construction.  First we define some auxiliary matrices $F_k$ for integer $k$ that will be 
used in the construction.  These can be 
thought of as infinite matrices, and we will use the notation $F_k^{s,t}$ to specify the $s$-by-$t$ matrix formed from 
the first $s$ rows and $t$ columns of $F_k$.  

\begin{definition}[$F_k$ matrices]
Let $k \in \mathbb{Z}$ and $i,j \in \mathbb{N}$.  The matrix $F_k$ is defined as
\[
F_k(i,j) = 
\begin{cases}
\binom{k-1}{j-i-1} & k >0 \\
(-1)^{j-i}\binom{-k-1+j-i}{-k-1} & k \le 0 \text{ and } i < j \\
(-1)^{i-j-k} \binom{i-j-1}{-k} & k \le 0 \text{ and } i \ge j
\end{cases}
\]
\end{definition} 
Notice that in each case the $(i,j)$ entry only depends on 
the difference $i-j$, thus each $F_k$ matrix is Toeplitz.  When $k>0$ we see that $F_k(i,j)=0$ whenever
$i \ge j$ meaning that these matrices are upper triangular.  When $k=0$ the definition simplifies 
to $F_0(i,j)=\binom{-1}{i-j}$, thus $F_0$ is lower triangular with ones on the 
main diagonal.  

To get a better idea where the $F_k$ come from, consider an extended Pascal's triangle where the upper and lower
indices begin from $-1$.  In the following table, the entries are binomial coefficients where upper indices label the rows, 
lower indices label the columns.
\begin{center}
 \begin{tabular}{c|cccccc}
 &-1 &0&1&2&3&4 \\
 \hline
 -1 & 0 &1 & -1 & 1 & -1&1\\
 0 & 0 &1 & 0 & 0 & 0&0\\
 1 & 0 &1& 1 &0 & 0&0 \\
 2& 0 & 1 & 2 & 1& 0&0\\
 3& 0 & 1 &3 & 3&1&0 \\
 4& 0 & 1 & 4 & 6 & 4 & 1
 \end{tabular}
 \end{center}
 The matrix $F_k$ for $k >0$ is the Toeplitz matrix whose first row is given by the row of Pascal's triangle indexed by 
 $k-1$, and whose first column is all zero.  For $k < 0$, up to signs, $F_k$ is a Toeplitz matrix whose first 
 column is given by the column of Pascal's triangle indexed by $-k$ and whose first row is given by the $-k-1$ column 
 of Pascal's triangle, starting from the row indexed by $-k-1$.

Using the $F_k$ we can now construct $A_r$.  
\begin{definition}[$A_r$ matrix]
For $r \ge 1$ let $A_r$ be a matrix of size $\binom{r+1}{2}$ defined as
 \[
  A_r = 
  \begin{bmatrix}
    F_0^{r,r} & F_{-1}^{r,r-1} & F_{-2}^{r,r-2} & \cdots & F_{-r+1}^{r,1} \\
    F_1^{r-1,r} & F_0^{r-1,r-1} & F_{-1}^{r-1,r-2} &  & F_{-r+2}^{r-1,1} \\
    F_2^{r-2,r} & F_1^{r-2,r-1} & F_0^{r-2,r-2} &  & F_{-r+3}^{r-2,1} \\
    \vdots &  &  & \ddots & \vdots \\
    F_{r-1}^{1,r} & F_{r-2}^{1,r-1} &\cdots  & \cdots & F_0^{1,1} \\
  \end{bmatrix}
\]
\end{definition}

The size of $A_r$ is clearly $\binom{r+1}{2}$.  
That $A_r$ is a Hadamard factorization of the identity and has rank $r$ will be shown in the next claims.  

We first show that $A_r$ is a Hadamard factorization of the identity.  This follows from the fact that $F_0$ is 
lower triangular and that in the above extended Pascal's triangle for $k > 0$ the row indexed by $k-1$ and 
column indexed by $k$ are disjoint.
\begin{claim}
$A_r$ is a Hadamard factorization of the identity.
\end{claim}

\begin{proof}
The diagonal entries of $A_r$ are $1$ as desired.  To show that $A_r(i,j)A_r(j,i)=0$ 
for $i \ne j$, it suffices to show that $F_k(i,j)F_{-k}(j,i)=0$ for each $k$.  This clearly holds for $k=0$ as 
$F_0$ is lower triangular.  Now suppose $k >0$.  If $i \ge j$ then $F_k(i,j)=0$ thus in this case we are also 
fine.  In the case $j > i$ we have
\[
|F_k(i,j) | | F_{-k}(j,i)| = \binom{k-1}{j-i-1} \binom{j-i-1}{k}=0 \enspace .
\]
The second term is zero for $j-i \le k$ while the first term is zero for $j-i \ge k+1$, thus the product is always zero.
\end{proof}
In fact, $A_r$ has the stronger property that exactly one of $A_r(i,j), A_r(j,i)$ is zero 
for $i \ne j$.  

The following claim is the key to prove $\rk(A_r) \le r$.
\begin{claim}
\label{claim:recurrence}
For any $k \in \mathbb{Z}$ and $i,j \in \mathbb{N}$
\[
F_k(i,j)=F_{k-1}(i,j) + F_{k-1}(i+1,j) \enspace .
\]
\end{claim}

\begin{proof}
We break the proof into $3$ cases depending on the value of $k$.  

\paragraph{Case 1 : $k >1$}
This case follows from the binomial addition formula
\begin{align*}
F_k(i,j)=\binom{k-1}{j-i-1}&=\binom{k-2}{j-i-1} + \binom{k-2}{j-i-2} \\
&=F_{k-1}(i,j) + F_{k-1}(i+1,j) \enspace.
\end{align*}

\paragraph{Case 2: $k=1$} In this case we use the symmetry identity together with binomial addition 
formula.  
\begin{align*}
F_1(i,j)=\binom{0}{j-i-1} = \binom{0}{i-j+1}&=\binom{-1}{i-j} + \binom{-1}{i-j+1} \\
& = F_0(i,j)+F_0(i+1,j) \enspace.
\end{align*} 

\paragraph{Case 3: $k \le 0$}
First consider the case $i \ge j$.  Then again by the binomial addition formula
\begin{align*}
F_{k}(i,j)&=
(-1)^{i-j-k}\binom{i-j-1}{-k} \\
&=(-1)^{i-j-k} \left(-\binom{i-j-1}{-k+1}+\binom{i-j}{-k+1} \right) \\
&=(-1)^{i-j-k+1}\binom{i-j-1}{-k+1}+(-1)^{i-j-k+2}\binom{i-j}{-k+1} \\
&=F_{k-1}(i,j) + F_{k-1}(i+1,j) \enspace .
\end{align*}

Finally, consider the case $i<j$.  This case requires some care as it could be that $i+1=j$.  For $k<0$, however, 
notice that the two formulas defining $F_k$ agree when $i=j$.  The first gives
$(-1)^{j-i}$ and the second $(-1)^{i-j-k} (-1)^{-k}=(-1)^{j-i}$.  Thus when $k<0$ and $i=j$ the two 
formulas in the definition are consistent.  As we are in Case 3, we are safe expressing $F_{k-1}(i+1,j)$
using the formula for $i < j$ as $k \le 0$.  
\begin{align*}
F_k(i,j)&=(-1)^{j-i}\binom{-k-1+j-i}{-k-1} \\
&=(-1)^{j-i} \left( \binom{-k+j-i}{-k} - \binom{-k+j-i-1}{-k}\right) \\
&=(-1)^{j-i} \binom{-k+j-i}{-k} + (-1)^{j-i-1}\binom{-k+j-i-1}{-k} \\
&=F_{k-1}(i,j) + F_{k-1}(i+1,j) \enspace. \\
\end{align*}
\end{proof}

\begin{claim}
The rank of $A_r$ is $r$. 
\end{claim}

\begin{proof}
The rank of $A_r$ is at least $r$ as the submatrix $F_0^{r,r}$ has rank $r$.  
\claimref{claim:recurrence} 
shows that all rows of $A_r$ can be expressed as linear combinations of the first $r$ rows, thus also $\rk(A_r) \le r$.
\end{proof}

\section{Conclusion}
We have shown a family of matrices $A_r$ of rank $r$ that are Hadamard factorizations of the identity of size 
$\binom{r+1}{2}$.  This construction is optimal in terms of rank, up to a constant factor, and answers an open question 
of Klauck and Wolf \cite{KlauckWolf13} and Friesen and Theis \cite{FriesenTheis13}.  

The original question of Dietzfelbinger et al.\ \cite{DHS96} remains open, however, as they specifically ask for the 
construction of a {\em boolean} matrix.  Our matrix has entries that are (positive and negative) 
integers.  Currently, the best known separation between rank and fooling set size for a boolean matrix is due to 
Theis \cite{Theis11}, who shows a boolean Hadamard factorization of the identity of size $n$ and rank 
$n^{\log(4)/\log(6)}$.  A potentially easier question would be to find a rank optimal Hadamard factorization of the 
identity that has only nonnegative entries.  Finally, the construction of Friesen and Theis in the finite field case 
gives matrices that are circulant.  How small can the rank over the reals be of a circulant Hadamard factorization of 
the identity?

\section*{Acknowledgements}
We would like to thank Dirk Oliver Theis for sharing this problem in the open problem session of the 
Dagstuhl seminar 13082 ``Communication Complexity, Linear Optimization, and lower bounds for the nonnegative 
rank of matrices''.  
We would also like to thank Andris Ambainis, Hartmut Klauck, and Dirk Oliver Theis for encouragement and helpful 
conversations during the course of this work.


\end{document}